\documentclass{svproc}
\usepackage{url}
\usepackage{enumerate}
\usepackage{amsmath}
\usepackage{stix}

\usepackage{graphicx}

\begin{document}
\mainmatter              
\title{Relative binary and ternary 4D velocities \\ in the Special Relativity
in terms of manifestly \\ covariant Lorentz transformation}
\titlerunning{Relative binary and ternary 4D velocities in terms of covariant Lorentz transformation}  
%
\author{Grzegorz Marcin Koczan}
\authorrunning{Grzegorz M. Koczan, GOL 2021 conference in memoriam Z. Oziewicz, Mexico} 
%
\tocauthor{Grzegorz Marcin Koczan}
\institute{Warsaw University of Life Sciences (PhD),\\
Faculty of Physics at the University of Warsaw (MSc)\\
\email{grzegorz\_koczan1@sggw.edu.pl, grzegorz.koczan@fuw.edu.pl}
}

\maketitle              

\begin{abstract}
Zbigniew Oziewicz was a pioneer of the 4D space-time approach to covariant relative velocities. In the 1980s (according to private correspondence) he discovered two types of 4D relative velocities: binary and ternary, along with the rules for adding them. They were first published in conference materials in 2004, and the second time in a peer-reviewed journal in 2007. These physically logical and mathematically precise concepts are so subtle that Oziewicz's numerous preprints have yet to receive the recognition they deserve. 

This work was planned to be a more review, but a thorough review of the little-known results was made in an original synthetic manner with numerous generalizations. The work presents the Oziewicz--Świerk--Bolós (and Matolcsi or Bini--Carini--Jantzen) binary relative velocity and the Oziewicz--Ungar--Dragan (also Celakoska--Chakmakov--Petrushevski on the basis of Urbantke, as well as Wyk) canonical ternary relative velocity. The Einstein--Oziewicz and Einstein--Minkowski relative velocities, which are a four-dimensional proper generalizations of Einstein relative velocity, also have a ternary character. This also applies to Oziewicz--Minkowski relative velocity, which is a time-like equivalent (generalization) of the canonical ternary velocity. It turns out that the Lorentz transformation of velocity itself is already ternary anchored. This fact is explicitly revealed by the manifestly covariant Lorentz transformation of velocity, which is the main tool of the work.

\end{abstract}
\textit{The fundamental property of velocities is the ability to add them.}
\begin{flushright}
Zbigniew Oziewicz (18 November 2020, online conversation)
\end{flushright}

\subsubsection{Introduction}

One of the pillars of Special Relativity (SR) is Einstein's principle of relativity or the related postulate of the Lorentz covariance of physics equations. Unfortunately, the postulate of covariance is often incorrectly called invariance (\cite{Herr,Ivezic} vs. \cite{Fahnline}).  Although the covariance requirement is widely acknowledged and emphasized, fundamental equations such as the Lorentz transformation and velocities composition are relieved of this requirement. This means that the mathematical formulas that describe them do not have a manifestly covariant form based on four-vectors. On the one hand, this condition results from the lack of popularization of the results of the manifestly covariant Lorentz transformation obtained in the last century by Bażański in 1965  \cite{Bazanski}, Krause in 1977 and 1978 \cite{Krause 1,Krause 2} and moreover by Fahnline in 1982 \cite{Fahnline}. On the other hand, this state of science is due to a misunderstanding and ignorance of the dependence of the Lorentz transformation on some additional direction in space-time. This direction, in the form of a time-like vector, appears explicitly in the aforementioned works by Krause and Fahnline, but is exposed to erroneous censorship in terms of an Æther that violates the principle of relativity. In this way, we have come to the bizarre apparent contradiction that the covariant Lorentz transformation breaks the basis on which it is based -- the principle of relativity -- due to the dependence on the distinguished time-like vector.

The explanation of this paradox and the further development of subtleties of a similar nature is provided by the research program of Professor Zbigniew Oziewicz  (see e.g. \cite{Oziewicz 3,Oziewicz Page}). This research program will be called Ternary Special Relativity (TSR) here, as it is based on the realization of the ternary nature of relative velocities in Special Relativity.  It should be emphasized that TSR is a developmental part of the SR and does not constitute an alternative content to the SR, but deepens its understanding in a new conceptual language. The concept of TSR was born thanks to the application of the methodology of the category theory to the kinematics of SR.  In simple terms, it was about morphisms and isomorphisms (called arrows) generated by Lorentz boosts or relativistic composition of velocities \cite{Oziewicz 2}. Although the graphs of the ternary relation of motion (\cite{Oziewicz 2}:Fig.1, \cite{Oziewicz 1}:Fig.1--3, \cite{Dragan}:Fig.3.2) look like an example of a category theory approach, this work is written in 4D covariant algebra. However, both in the covariant approach and in the category theory, the most overt independence of the choice of the coordinate system is postulated.  It should be noted, however, that TSR in a way separates the choice of coordinates from the choice of the observer -- which in a certain sense doubles the concept of relativity. Therefore, if we have Cartesian coordinates, the observer and the body (or the observer and two bodies), then we can speak of a ternary system needed to study ternary relations of motion. 

The discovery of the ternary nature of the relativistic composition of velocities raised in Oziewicz a natural question about the possibility of breaking the hegemony of the Lorentz group down to the level of a particular kind of abstract Æther \cite{Oziewicz eter}. I believe that in the correspondence with Oziewicz it was possible to work out a certain interpretative consensus on this matter. As part of this consensus, I proposed in correspondence with Oziewicz the term Oziewicz's Æther for each of the two 2D planes distinguished by a Lorentz boost.  Oziewicz's spacetime-like Æther is a 2D hyperplane in which Lorentz boost works, while Oziewicz's space-like Æther is a 2D surface, which Lorentz boost does not change. Both such 2D Oziewicz's Æthers are not any material distinguished frames of reference, but they constitute a certain subtle distinguishing, which is introduced by the mere Lorentz boost.  Therefore, according to the author, TSR does not break the Lorentz group and the SR itself.  Nevertheless, Oziewicz himself expressed some doubts in this respect  \cite{Oziewicz 2,Oziewicz 4}. These doubts were mainly based on the lack of  associativity for relativistic velocities addition and the multitude of relative zero velocities.  The lack of  associativity of standard relativistic velocities addition is an indisputable fact \cite{Ungar,Kocik}, but it concerns the pseudo-group (groupoid or loop) of velocities, and not the entire Lorentz group, which includes, next to Lorentz's boosts, also rotations. Other aspects of non-associativity in physics were cited in \cite{Sbitneva}.   

This does not mean that the Thomas-Wigner rotation accompanying composition of non-parallel velocities depreciates the mathematical structure for composition velocities.  Oziewicz called this properly interpreted structure a relativistic groupoid \cite{Oziewicz 2}, which he argued (among others) by the aforementioned multitude of relative zero velocities. Others refer to it as the Lobachevsky--Einstein velocity space or a gyrogroup.  Investigating this structure, however called, does not constitute a breach of the Lorentz group, whether or not we agree with the somewhat exotic multi-zero argument or other argument.  In other words, according to the author, the TSR should be treated as an original and developmental, but orthodox approach to the SR.  Zbigniew Oziewicz, the creator of this program, was looking for new solutions and perhaps he would not literally agree with the above opinion. Undoubtedly, Oziewicz introduced a new concepts within TSR and SR, and whether or not we treat TSR as going beyond the SR or not is of secondary importance. As already suggested, the author of the article classifies TSR as SR without harming any of them.

The only more widely known competitive approach to TSR (including the Oziewicz relativistic groupoid) is the aforementioned Ungar's gyrogroup. However, according to the author, the idea of a gyrogroup is less innovative and less fundamental than TSR. The gyrogroup is based on the hyperbolic geometry of the relativistic Lobachevsky--Einstein velocity space (see appendix in \cite{RinP}), the foundations of which were already described in 1911 by Vari\'cak, and further developed e.g.  by Fock in the 1960s.  Moreover, the gyrogroup did not recognize such fundamental concepts as the covariant binary and ternary relative velocities. Strictly speaking, Ungar in his very extensive research introduced the equivalent of ternary velocity, but did not recognize its essential role and ternary character.  Ungar's greatest discovery, as Oziewicz himself emphasized, is the discovery of the lack of associativity for the relativistic addition of velocities \cite{Ungar}. Both Ungar and Oziewicz proposed some concepts of a weaker version of associativity as part of their own approach to the velocities composition problem. 

It's time to define the basics of TSR in 4D covariant vectors algebra.  Consider two material points to which 4D absolute velocities are assigned, respectively: $\nu=\{\nu^{\alpha}\}$ and $\mu=\{\mu^{\alpha}\}$. These are time-like vectors directed to the future, which in the signature $+---$ are normalized in a positive sense to the speed of light $\nu \circ \nu=\nu^{\alpha}\nu_{\alpha}=+c^2=\mu \circ \mu$. These will be referred to here as Minkowski 4D absolute velocities or simply Minkowski velocities. Minkowski velocity is absolute in the sense that it does not distinguish other body than the body it describes. The concept of absolute velocity only makes sense in space-time (4D), while in ordinary space (3D without time and without ether) it doesn't make sense anymore. The name four-velocity is commonly used for them, which Oziewicz did not like. Oziewicz, following Minkowski, thought that the $\nu^{\alpha},\mu^{\alpha}$ objects describe material points (material bodies) as such, and do not constitute their velocities. In a way, he was right, but in the sense of properly scaled spatial components, Minkowski time-like vectors do describe velocities.   

On the basis of Minkowski velocities of two material points, it is possible to build a general concept of relative velocity -- first of the binary type:
\begin{equation}
\label{binary}
    \mathcal{B}^{\alpha}(\nu, \mu)=\lambda_2(\nu,\mu) \mu^{\alpha}-\lambda_1(\nu,\mu) \nu^{\alpha}.
\end{equation}
Binary velocity, like any relative velocity, is a certain subtraction of component velocities. Such subtraction in relativism requires at least one non-trivial “linear” combination coefficient. In the case of the canonical Oziewicz--Świerk--Bol\'os binary velocity $\omega^{\alpha}$ only the $\lambda_2$ coefficient is non-trivial.  Any binary velocity should be a space-like vector normalized (in a negative sense) to a speed lower than the speed of light ($\mathcal{B}^{\alpha}\mathcal{B}_{\alpha}=-B^2,0\leq B<c$). 

If material points with Minkowski velocities $\nu^{\alpha}, \mu^{\alpha}$ are observed by an observer with Minkowski velocity $\sigma^{\alpha}$, then the relative velocity of the ternary type can be considered: 
\begin{equation}
\label{general ternary}
    \mathcal{T}^{\alpha}(\sigma, \nu, \mu)=f_2(\sigma, \nu,\mu) \mu^{\alpha}-f_1(\sigma, \nu,\mu) \nu^{\alpha}-f_0(\sigma, \nu,\mu) \sigma^{\alpha}.
\end{equation}
The point is that we want to relate (transfer) $\mu^{\alpha}$ velocity relative to $\nu^{\alpha}$ to the $\sigma^{\alpha}$ system. In practice, such a reference requires orthogonality ($\mathcal{T}^{\alpha}\sigma_{\alpha}=0$), which partially explains the difference of $\mathcal{T}^{\alpha}$ compared to $\mathcal{B}^{\alpha}$. The fulfillment of this condition can be guaranteed automatically by a “linear” combination of canonical binary velocities:
\begin{equation}
    \mathcal{T}^{\alpha}=g_2 \omega^{\alpha}_2(\sigma,\mu)-g_1 \omega^{\alpha}_1(\sigma,\nu)=g_2 \frac{c^2}{\sigma\circ\mu}\mu^{\alpha}-g_1 \frac{c^2}{\sigma\circ\nu}\nu^{\alpha}-g_0\sigma^{\alpha}.
\end{equation}
Such a notation only has the coefficients $g_0,g_1,g_2$ scaled with respect to $f_0,f_1,f_2$, which simplifies the orthogonality condition to $g_0=g_2-g_1$.  The canonical ternary velocity of Oziewicz--Ungar--Dragan $\xi^{\alpha}_{(\sigma)}$, however, has equal two coefficients of the first kind  $f_1=f_2$, not the second kind  ($g_1$ and $g_2$). Of course, any ternary velocity should come down to the binary velocity for  $\sigma^{\alpha}=\nu^{\alpha}$.

In further research pseudo-binary velocities  also considered.  A pseudo-binary velocity is a ternary velocity in which a principal coefficient is simple, i.e. unitary (e.g. $g_1=1$ or $g_2=1$), or is a function of only two, not three, Minkowski velocities.  The pseudo-binary velocity will also be a velocity similar to (\ref{binary}), at least one coefficient of which will depend on the third Minkowski velocity  $\sigma^{\alpha}$. The term pseudo-binarity refers to ternarity that is not fully manifested, and here it has nothing to do with the change of signs on reflections (as in pseudo-tensors). The best example of a pseudo-binary velocity from further research will be the cross velocity $\varpi^{\alpha}$. Another example would be axial velocity  $\pi^{\alpha}$, which is a 4D generalization of the relative velocity of Fernándeuz--Guasti (and the author). The pseudo-binarity requirements are also met (in sense of norm) by the velocity of Einstein--Oziewicz $\varepsilon^{\alpha}$, consideration in this work, which generalizing to 4D the Einstein relativistic composition of velocity. We assume that pseudo-binary velocity, like binary and canonical ternary, is a space-like vector  ($\varpi^{\alpha}\varpi_{\alpha}<0,\  \pi^{\alpha}\pi_{\alpha}<0, \  \varepsilon^{\alpha}\varepsilon_{\alpha}<0$).  A description of further research of pseudo-binary follows after the conclusion of this work.

An interesting fact is the existence of a time-like ($\beta^{\alpha}\beta_{\alpha}=+c^2>0$) relative velocity of Einstein--Minkowski $\beta^{\alpha}$ related to $\varepsilon^{\alpha}$ and with the covariant Lorentz transformation.  This velocity is therefore a covariant Lorentz transformation of Minkowski velocity. It resembles ternary velocity or pseudo-binary velocity, but is a time-like vector.

Some of the following precise conditions are imposed on the relative covariant velocities (binary, ternary, pseudo-binary) occurring in this work: 
\begin{enumerate}[i.]
     \item Compliance of velocities composition with the Lorentz transformation (Einstein composition) for the component parallel to the boost velocity (in the 3D sense).
    \item Compatibility of velocities composition with the Lorentz transformation (Einstein composition) only in the case of parallel velocities (in the 3D sense).
    \item Corresponding compatibility of the covariant active Lorentz boost of velocity $\nu^{\alpha}$ performed with some (ternary) relative velocity (agreement in the active sense).
    \item Equality of the 4D relative velocity with the value of some covariant passive Lorentz boost of velocity $\mu^{\alpha}$ (agreement in the resting sense).
    \item Orthogonality to Minkowski velocity  reference ($\mathcal{B}^{\alpha}\nu_{\alpha}=0$ or $\mathcal{T}^{\alpha}\sigma_{\alpha}=0$).
    \item Compatibility with 3D velocity corresponding to $\mu^{\alpha}$ when $\nu^{i}=0$.
    \item The triviality of one of the main coefficients ($f_1=1$ or $f_2=1$ or $g_1=1$ or $g_2=1$).
    \item  Antisymmetry, i.e. the equality of the main coefficients  ($f_1=f_2$). 
\end{enumerate}

The work describes the individual relative velocities generally in separate sections. These sections refer to the above conditions, identifying specific ones and, if necessary, further specifying the determining conditions.  Sections 1 and 2 were intended to be essentially a review, but ultimately have some authorial input. Whereas section 3 is more original. However, the covariant Lorentz transformation (including indirectly $\beta^{\alpha}$ velocity) has been found in several sources. On the other hand, the four-dimensional velocity $\varepsilon^{\alpha}$ has not been strictly recognized in the literature. Probably even Oziewicz wrongly identified the velocity of Einstein type $\varepsilon^{\alpha}$ (\cite{Oziewicz 0}:(7.1), \cite{Oziewicz 1.5}:24--25) with the ternary velocity $\xi^{\alpha}$ (\cite{Oziewicz 4}:(9.2),(16.10)). In addition to determining relative velocities, the article deals with the issue of their addition in accordance with the article's motto, which is a quote of prof. Zbigniew Oziewicz.  

\section{The binary 4D relative velocity of Oziewicz--Świerk--Bol\'os and its transitive-type addition with ternary properties}

If the conditions v. and vi. are set for the binary velocity (\ref{binary}), then we can easily obtain the canonical binary velocity of Oziewicz--Świerk--Bol\'os (see \cite{Oziewicz 1,Oziewicz 2,Swierk,Bolos,APP A} and Fig.\ref{binaryfig}): 
\begin{equation}
\label{binary OB}
    \omega^{\alpha}(\nu, \mu)=\frac{c^2}{\nu\circ\mu} \mu^{\alpha}-\nu^{\alpha}=:(\mu\dsub\nu)^{\alpha}.
\end{equation}
Condition vi. could be replaced by condition vii. in the form $\lambda_1=1$. On the other hand, the $\lambda_2$ coefficient is equal to the reciprocal of the invariant Lorentz--Oziewicz gamma coefficient for the binary velocity \cite{RinP}:
\begin{equation}
    \gamma_{12}=\gamma_{21}=\frac{\nu\circ\mu}{c^2}=\frac{1}{\sqrt{1+\omega\circ\omega/c^2}}=\frac{1}{\lambda_2}.
\end{equation}
As it is easy to see the coefficients $\gamma_ {12}, \lambda_2$ are dependent on the scalar square of the space-like binary velocity  $\omega\circ\omega<0$. 

\begin{figure}[h!]
\centering
\includegraphics[width=8cm]{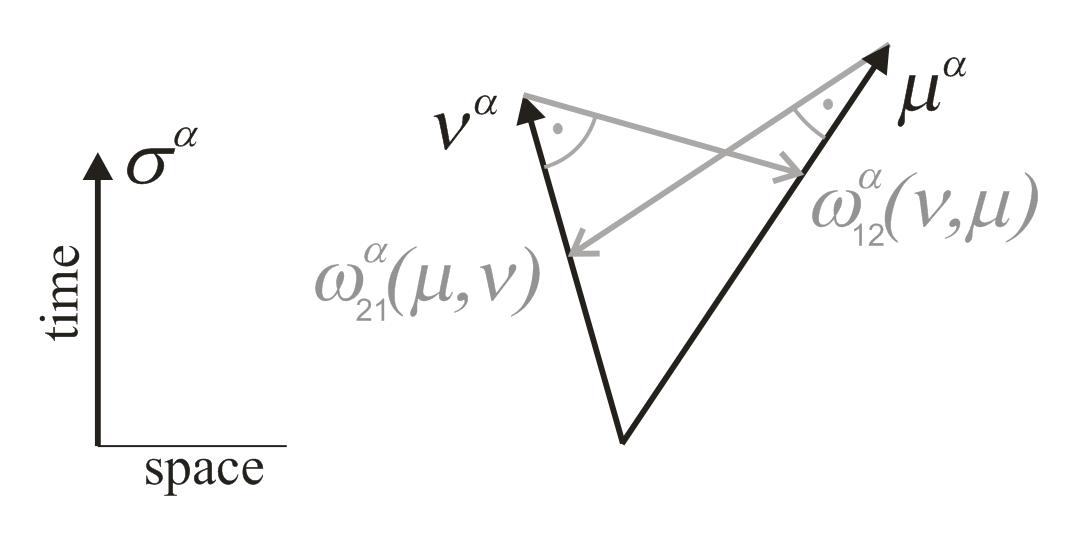}
\caption{Simple Bini--Carini--Jantzen's geometric construction \cite{Bini,poster 2} of 4D relative velocity $\omega^{\alpha}_{12}(\nu,\mu)$ called Oziewicz--Świerk--Bol\'os binary velocity \cite{Oziewicz 1,Oziewicz 2,Swierk,Bolos,APP A}, along with the inverse binary velocity $\omega^{\alpha}_{21}(\mu,\nu)$, in the Minkowski space-time diagram. The constructed relative velocities of two bodies with absolute Minkowski velocities $\nu^{\alpha}$ and $\mu^{\alpha}$ are determined by the appropriate orthogonal line segments to these Minkowski velocities. Both binary velocities have equal values, but indicate different directions in space-time.}
\label{binaryfig}
\end{figure}

It turns out that the conditions i. and ii. are not fulfilled here, because for parallel velocities in the sense of 3D, the binary velocity is $\gamma_{v}=1/\sqrt{1-v^2/c^2}$ times greater (wherein $\nu^i=\gamma_vv^i$). This means that canonical binary velocity can be written as  \cite{APP A}:
\begin{equation}
    \omega^{\alpha}=(\omega^0,\vec{\omega})=(\gamma_v \vec{w}\cdot\vec{v}/c,\gamma_v\vec{w}),
\end{equation}
where $\vec{w}=\vec{\omega}/\gamma_v$ is 3D binary velocity \cite{APP A} or jet velocity \cite{RinP}. The $\vec{w}$ velocity satisfies the full version of condition i. (not the limited condition ii.). In further research the 4D generalization of pseudo-binary velocity $\vec{w}$ will be given, called the cross velocity $\varpi^{\alpha}$.

Binary velocity in the form of a 4D vector (\ref{binary OB}) appeared in 1988 in Świerk's master thesis \cite{Swierk}, the originator and promoter of which was Oziewicz. Nevertheless Oziewicz first published the canonical binary velocity only in 2004 \cite{Oziewicz 1} (perhaps it was blocked by reviewers \cite{Oziewicz Reviewers}). A year later, the same velocity was published by Bol\'os \cite{Bolos}, but for General Relativity. In a similar context, binary velocity (without a specific name) appears in the newer work of \cite{Korzynski}. On the other hand, in older works, such as Hestenes \cite{Hestenes} from 1974, relative velocity appears in the form of a bivector, not a vector.   As evidence of the fact that the important velocity vector $\omega^{\alpha}$ is still not widely known, let's consider the article written by a recognized relativist \cite{Abramowicz}. This article comes close to canonical binary velocity but does not define its proper normalization. Despite the failure to fulfill the condition i. the space-time normalization is correct and is equal to the norm for Einstein composition $\omega^{\alpha}\omega_{\alpha}=\varepsilon^{\alpha}\varepsilon_{\alpha}=-|\vec{\mathcal{E}}|^2$ (see Statement 2 in \cite{APP A}). The author also came (by chance) to binary velocity in 2019/2021 in the relativistic equation of motion of a body with a variable own mass \cite{RinP}. The idea of binary velocity apparently also appears in the 1995 work by Bini, Carini and Jantzen \cite{Bini}. A distractor to find this work was to refer to the General Relativity (also \cite{Bolos,Korzynski}), when the essence of the problem belongs to the Special Relativity. In this context, binary velocity was introduced two years earlier (1993) by Matolcsi \cite{Mat0}, which he also described in his next work with Goher \cite{Mat}. Binary velocity also appears in the newer \cite{Celakoska 4} work.

The binary relative velocity makes it possible to express the Minkowski velocity of the second material point in terms of the Minkowski velocity of the first point: 
\begin{equation}
\label{binary plus}
    \mu^{\alpha}(\omega, \nu)=\frac{\omega^{\alpha}+\nu^{\alpha}}{\sqrt{1+\omega\circ\omega/c^2}}=:(\omega\ \langle + | \nu)^{\alpha}.
\end{equation}
This equation results from a simple transform (\ref{binary OB}) and the resulting square $\omega\circ\omega=c^6/(\sigma\circ\nu)^2-c^2$. 
The operations $\dsub$ (\ref{binary OB}) and $\langle +|$ (\ref{binary plus}) are not internal operations. The operation $\dsub$ is specified on Minkowski time-like velocities, but its value is a space-like binary velocity. And $\langle+|$ works between binary and Minkowski velocity, and the value is Minkowski velocity.

Now consider composition of the binary velocities themselves. Some kind of transitivity can be used in the definition of such composition:
\begin{equation}
    \big(\omega_{10}(\nu,\sigma)\boxplus_{B}^{\sigma}\omega_{02}(\sigma,\mu)\big)^{\alpha}:=\omega^{\alpha}_{12}(\nu,\mu),
\end{equation}
where: $\sigma=\{\sigma^{\alpha}\}$ -- the Minkowski velocity of intermediate (third) material point; $10,01,12$ -- subscripts distinguish between three binary velocities and indicate a reference relation to pairs of material points (analogous to Fig. \ref{binaryfig}). The designation of the addition of binary velocities $\boxplus_{B}^{\sigma}$ suggests its dependence on an additional parameter, which was select Minkowski velocity of the material middle point. Equivalently, a different point could be selected or the inverse velocity $\omega^{-1}_{10}(\nu,\sigma):=\omega_{01}(\sigma,\nu)\neq-\omega_{10}(\nu,\sigma)$ alternately (compare with Fig. \ref{binaryfig}). Oziewicz used the latter choice. 
\begin{theorem}
Composition (adding) two binary velocities $\omega_{10}^{\alpha}$ and $\omega_{02}^{\alpha}$, meeting the conditions of composition $\sigma\circ\omega_{10}/c^2=\gamma_{10}^{-1}-\gamma_{10}$ and $\sigma\circ\omega_{02}=0$, is a ternary expression depending on these two velocities and on the Minkowski velocity $\sigma^{\alpha}$ of the middle material point and is given by the formula:
\begin{equation}
    (\omega_{10}\boxplus_{B}^{\sigma}\omega_{02})^{\alpha}:=\omega_{10}^{\alpha}+\frac{\omega_{02}^{\alpha}+\sigma^{\alpha}}{\gamma_{10}-\omega_{10}\circ\omega_{02}/c^2}-\frac{\sigma^{\alpha}}{\gamma_{10}},
\end{equation}
where the invariant Lorentz--Oziewicz factor has the form  $\gamma_{10}=1/\sqrt{1+\omega_{10}\circ\omega_{10}/c^2}=\gamma_{01}$.
\end{theorem}

\begin{proof}
Using the relation (\ref{binary plus}) we calculate the Minkowski velocity $\mu^{\alpha}$ of the third material point (No. 2) adding to the relative binary velocity $\omega_{02}^{\mu}$ Minkowski velocity $\sigma^{\alpha}$: 
\begin{equation}
    \mu^{\alpha}=(
    \omega_{02}\ \langle + | \ \sigma)^{\alpha}=\frac{\omega_{02}^{\alpha}+\sigma^{\alpha}}{\sqrt{1+\omega_{02}\circ\omega_{02}/c^2}}=\gamma_{02}(\omega_{02}^{\alpha}+\sigma^{\alpha}).
\end{equation}
From the transformation of the analogous relation for the first (No. 1) and the second material point (No. 0), it is possible to calculate the Minkowski velocity $\nu^{\alpha}$ of the first material point:
\begin{equation}
    \nu^{\alpha}=\frac{1}{\gamma_{10}}\sigma^{\alpha}-
    \omega_{10}^{\alpha}=:(\sigma\  |-\rangle \
    \omega_{10})^{\alpha},
\end{equation}
where, by the way, the subtraction operation of the inverse (in relation to $\sigma^{\alpha}$) binary velocity $\omega_{10}(\nu,\sigma)=\omega^{-1}_{01}(\sigma,\nu)$ was defined.
The scalar product of the first (No. 1) and last Minkowski (No. 2) velocities is: 
\begin{equation}
    \nu\circ\mu=\gamma_{02}(\gamma_{10}^{-1}(c^2+\sigma\circ\omega_{02})-\sigma\circ\omega_{10}-\omega_{10}\circ\omega_{02})=\gamma_{02}(\gamma_{10}c^2-\omega_{10}\circ\omega_{02}),
\end{equation}
where the assumptions (conditions of composition) were used. 
Thus, the composition of the velocities, which is the binary velocity $\omega_{12}^{\alpha}$ takes the form:
\begin{equation}
    (\omega_{10}\boxplus_{B}^{\sigma}\omega_{02})^{\alpha}=\omega_{12}^{\alpha}=(\mu\dsub\nu)^{\alpha}=\frac{\sigma^{\alpha}+\omega_{02}^{\alpha}}{\gamma_{10}-\omega_{10}\circ\omega_{02}/c^2}-\frac{\sigma^{\alpha}}{\gamma_{10}}+\omega_{10}^{\alpha},
\end{equation}
which, with the accuracy of the order of writing, constitutes the thesis of the theorem. \qed
\end{proof}
To check compliance with Oziewicz's original formula, we can use the useful equation $\omega_{10}\circ\omega_{02}= -\gamma_{10}\omega^{-1}_{10}\circ\omega_{02}$ -- for which the result value must be additionally changed (remove the minus) due to other signatures $(+---)$ vs $(-+++)$.

Due to the fact that the composition of binary velocities depends on the middle material point, and not the extreme point or inverse binary velocity, it is possible to refine the associativity of velocities composition in the sense of Oziewicz, based on transitivity:
\begin{equation}
    \Big(\omega_{10}(\nu,\sigma)\boxplus_{B}^{\sigma}\omega_{02}(\sigma,\mu)\Big)\boxplus_B^{\mu}\omega_{23}(\mu,\beta)=\omega_{10}(\nu,\sigma)\boxplus_{B}^{\sigma}\Big(\omega_{02}(\sigma,\mu)\boxplus_B^{\mu}\omega_{23}(\mu,\beta)\Big),
\end{equation}
where both sides are by definition equal  $\omega_{13}(\nu,\beta)=\{\omega^{\alpha}_{13}(\nu,\beta)\}$. The limited sense of such associativity consists in numerous constraints (conditions) for composition of  binary velocities and in the ternary nature of such composition (the dependence on the third parameter $\sigma^{\alpha}$ or $\mu^{\alpha}$).

\section{The ternary Oziewicz--Ungar--Dragan relative velocity and special case of its addition in the 4D Oziewicz--Einstein sense}

Already composition of binary velocities shows ternarity and dependence on the distinguished frame of reference. Additionally, the binary composition required uncomfortable orthogonality conditions:  the velocity $\omega_{02}^{\alpha}$ is othogonal to $\sigma^{\alpha}$, but velocities $\omega_{10}^{\alpha}, \omega_{12}^{\alpha}$ are ortgononal to $\nu^{\alpha}$.
It is therefore worth considering a velocities subset that will contain orthogonal velocities up to the one velocity only (e.g. to $\sigma^{\alpha}$ -- see Fig. \ref{ternaryfig}). Simply put, Oziewicz followed this path, introducing a relative ternary velocity \cite{Oziewicz 1,Oziewicz 3}. Above all, however, he referred to the Lorentz transformation (in the active sense iii., not resting sense iv.).

In this work, the ternary velocity will be introduced directly at the level of the covariant Lorentz transformation described in the next section. It should be known that the Lorentz boost depends on the bivector, which can be described by two Minkowski velocities (reference velocity $\sigma^{\alpha}$ and boost velocity $\zeta^{\alpha}$) or equivalently  Minkowski velocity ($\sigma^{\alpha}$) and Oziewicz velocity ($\omega^{\alpha}$ or $\xi^{\alpha}$) orthogonal to the first ($\omega\circ\sigma=0$ or $\xi\circ\sigma=0$):
\begin{equation}
    (\sigma \wedge \zeta)^{\alpha\beta}:=\sigma^{\alpha}\zeta^{\beta}-\sigma^{\beta}\zeta^{\alpha}=\sigma^{\alpha}\xi^{\beta}-\sigma^{\beta}\xi^{\alpha}=:(\sigma \wedge \xi)^{\alpha\beta}.
\end{equation}
A more natural parameterization seems to be Oziewicz space-like velocity (ternary $\xi^{\alpha}$ or binary $\omega^{\alpha}$), as it is normalized (with the signature minus) to the usual boost velocity value.

\begin{figure}[h!]
\centering
\includegraphics[width=8cm]{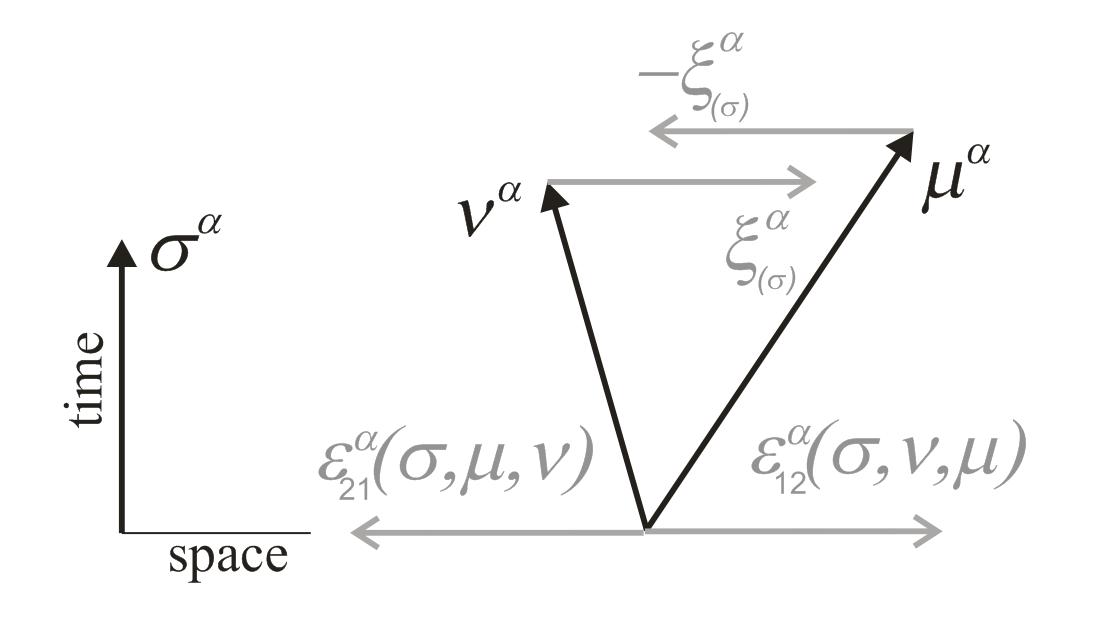}
\caption{Visualization of the Oziewicz--Ungar--Dragan ternary relative velocity $\xi^{\alpha}_{(\sigma)}$ and the Einstein--Oziewicz relative velocity $\varepsilon^{\mu}_{12}(\sigma,\nu,\mu)$, along with the inverse velocities   $-\xi^{\alpha}_{(\sigma)}$ and $\varepsilon^{\mu}_{21}(\sigma,\mu,\nu)$, on the space-time Minkowski diagram . Apart from the Minkowski absolute velocities $\nu^{\alpha}$ and $\mu^{\alpha}$ of two material points, the above relative velocities depend on the Minkowski velocity $\sigma^{\alpha}$ of the observer, what illustrates the relativistic property of ternarity. The Oziewicz--Ungar--Dragan ternary velocity $\xi^{\alpha}_{(\sigma)}$ is a space-like parameter (inverse velocity in another sense) of the active Lorentz boost transforming of $\nu^{\alpha}$ to $\mu^{\alpha}$. Whereas, the Einstein--Oziewicz velocity $\varepsilon^{\mu}_{12}$ is a space-like version of the velocity transformation $\mu^{\alpha}$ of passive Lorentz  boost to the system with the velocity  $\nu^{\alpha}$.}
\label{ternaryfig}
\end{figure}

Nevertheless, we will begin by defining a Minkowski type velocity $\zeta^{\alpha}$ that parameterizes the following active Lorentz boost (or equivalent passive inverse boost):
\begin{equation}
\label{zeta}
    La(\sigma,\zeta)\big[\{\nu^{\alpha}\}\big]:=\{\mu^{\alpha}\}\ \longleftrightarrow \ Lp(\sigma,\zeta)\big[\{\mu^{\alpha}\}\big]:=\{\nu^{\alpha}\}.
\end{equation}
Thanks to the Theorem \ref{Theorem inverse velocity} proved in the next section (see (\ref{inverse})), it is possible to provide a solution to the above equation taking into account the notations used here: 
\begin{equation}
  \zeta^{\alpha}=\frac{\sigma\circ\nu+\sigma\circ\mu}{c^2-\nu\circ\mu+2(\sigma\circ\nu)(\sigma\circ\mu)/c^2}(\mu^{\alpha}-\nu^{\alpha})-\frac{c^2-\nu\circ\mu-2(\sigma\circ\nu)^2/c^2}{c^2-\nu\circ\mu+2(\sigma\circ\nu)(\sigma\circ\mu)/c^2}\sigma^{\alpha}.
\end{equation}
Based on the scalar product:
\begin{equation}
  \zeta\circ\sigma=\frac{(\sigma\circ\nu)^2+(\sigma\circ\mu)^2+c^2\nu\circ\mu-c^4}{c^2-\nu\circ\mu+2(\sigma\circ\nu)(\sigma\circ\mu)/c^2},
\end{equation}
it is already possible to determine the canonical ternary Oziewicz velocity: $\xi^{\alpha}_{(\sigma)}=(\zeta\dsub\sigma)^{\alpha}$ as relative velocity (in binary sense) $\zeta^{\alpha}(\sigma,\nu,\mu)$ with respect to $\sigma^{\alpha}$:
\begin{equation}
  \xi^{\alpha}_{(\sigma)}=\frac{c^2\sigma\circ\nu+c^2\sigma\circ\mu}{(\sigma\circ\nu)^2+(\sigma\circ\mu)^2+c^2\nu\circ\mu-c^4}(\mu^{\alpha}-\nu^{\alpha})-\frac{(\sigma\circ\mu)^2-(\sigma\circ\nu)^2}{(\sigma\circ\nu)^2+(\sigma\circ\mu)^2+c^2\nu\circ\mu-c^4}\sigma^{\alpha}.
\end{equation}
This could be written a little shorter  (see \cite{APP A,RinP}):
\begin{equation}
  \xi^{\alpha}_{(\sigma)}(\sigma,\nu,\mu)=\frac{c^2\sigma\circ\nu+c^2\sigma\circ\mu}{(\sigma\circ\nu)^2+(\sigma\circ\mu)^2+c^2\nu\circ\mu-c^4}\Big(\mu^{\alpha}-\nu^{\alpha}-\frac{\sigma\circ\mu-\sigma\circ\nu}{c^2}\sigma^{\alpha}\Big).
\end{equation}
This velocity in this work is called Oziewicz--Ungar--Dragan ternary velocity (Fig. \ref{ternaryfig}), according to Oziewicz \cite{Oziewicz 1,Oziewicz 2,Oziewicz 3}, even though Ungar \cite{Ungar 2001,Ungar 2005} and Dragan \cite{Dragan,Dragan WS} only gave its 3D version (see simple interpretation by the laufer/bishop method  \cite{poster}), which is not explicitly ternary. Moreover, the Ungar approach was not about relative velocity (subtraction), but addition (“Einstein coaddition”).  In order to expose this ternarity, one can introduce the ternary subtraction of Minkowski velocities leading to the ternary velocity:
\begin{equation}
  (\mu\triangleminus_{\sigma}\nu)^{\alpha}:=\xi^{\alpha}_{(\sigma)}(\sigma,\nu,\mu).
\end{equation}
This does not contradict binary relationships between velocities $\zeta^{\alpha}$ and $\xi^{\alpha}_{(\sigma)}$:  
\begin{equation}
 \xi^{\alpha}_{(\sigma)}=(\zeta\dsub\sigma)^{\alpha}=\omega^{\alpha}(\sigma,\zeta)\longleftrightarrow \zeta^{\alpha}=(\xi_{(\sigma)}\  \langle + | \
    \sigma)^{\alpha}.
\end{equation}
Thus, the velocity $\zeta^{\alpha}$ is the time-like equivalent of the space-like velocity $\xi^{\alpha}_{(\sigma)}$, so it can be called the Oziewicz--Minkowski velocity.
On the other hand, the more direct relations (without $\zeta^{\alpha}$) are ternary (see Fig. \ref{ternaryfig}):
\begin{equation}
 \xi^{\alpha}_{(\sigma)}=(\mu\triangleminus_{\sigma}\nu)^{\alpha}\longleftrightarrow \mu^{\alpha}=:(\xi_{(\sigma)}\triangleplus_{\sigma}
    \nu)^{\alpha}.
\end{equation}

Unfortunately, all of the above operations are not internal operations. As for binary velocities, also for ternary velocities there should be a certain rule of their addition. In further research this problem will be attacked in general, and now Oziewicz interpretation related to Einstein velocity composition will be given.The key here is the relation of the Lorentz transformation in the form (\ref{zeta}). It means that the “orthodox” composition of Minkowski velocities $\zeta^{\alpha}$ and $\nu^{\alpha}$ is Minkowski velocity $\mu^{\alpha}$. Now it is enough to convert the time-like Minkowski velocities into relative space-like velocities (binary and ternary Oziewicz velocities). So, if we define the addition of ternary velocities in the sense of Oziewicz as follows: 
\begin{equation}
\label{ternarne1}
    \big(\xi_{(\sigma)01}(\sigma,\sigma,\nu)\boxplus_{O}^{\sigma}\xi_{(\sigma)12}(\sigma,\nu,\mu)\big)^{\alpha}:=\xi^{\alpha}_{(\sigma)02}(\sigma,\sigma,\mu)=\omega_{02}^{\alpha}(\sigma,\mu),
\end{equation}
it will be, in a sense, Lorentz boost:
\begin{equation}
\label{ternarne2}
    \Big((\nu\dsub\sigma)\boxplus_{O}^{\sigma}\big(\zeta(\nu\rightarrow\mu)\dsub\sigma\big)\Big)^{\alpha}:=(\mu\dsub\sigma)^{\alpha}.
\end{equation}
The sense of such Lorentz boost is developed in the following theorem. 
It is worth noting that ternary addition in the sense of Oziewicz cannot be associative, because one of the added velocities (as well as the resulting velocity) is ordinary binary velocity, not ternary (it depends twice on $\sigma$). 

\begin{theorem}[Oziewicz \cite{Oziewicz 3}]
Ternary velocities composition in the sense of Oziewicz $\boxplus^{\sigma}_{O}$ (adding ternary velocity $\xi^{\alpha}\equiv\xi^{\alpha}_{(\sigma)}$) in the material point system $\sigma^{\alpha}$ is the 4D version of Einstein's 3D velocity addition ($\oplus_0$) and has an analogous algebraic form:
\begin{equation}
    (\omega\boxplus_{O}^{\sigma}\xi)^{\alpha}=\frac{\omega^{\alpha}+\xi^{\alpha}+\frac{\gamma_{\xi}}{\gamma_{\xi}+1}\big((\omega\wedge\xi)\odot\xi\big)^{\alpha}/c^2}{1-\omega\circ\xi/c^2}=\frac{\omega^{\alpha}+\xi^{\alpha}+\frac{\gamma_{\xi}}{\gamma_{\xi}+1}\Big(\frac{\xi\circ\xi}{c^2}\omega^{\alpha}-\frac{\omega\circ\xi}{c^2}\xi^{\alpha}\Big)}{1-\omega\circ\xi/c^2}.
\end{equation}
In order to refer to the 3D version, it is enough to convert tensor products (antisymmetric $\wedge$ or internal $\odot$) to 3D vector products and to change to opposite signs of scalar products (transition from signature $+---$ to signature $+++$):  
\begin{equation}
\label{Einstein}
    \vec{v}\oplus_{0}\vec{W}=\frac{\vec{v}+\vec{W}+\frac{\gamma_{W}}{\gamma_{W}+1}(\vec{v}\times\vec{W})\times\vec{W}/c^2}{1+\vec{v}\cdot\vec{W}/c^2}=\frac{\vec{v}+\vec{W}+\frac{\gamma_{W}}{\gamma_{W}+1}\Big(\frac{\vec{v}\cdot\vec{W}}{c^2}\vec{W}-\frac{W^2}{c^2}\vec{v}\Big)}{1+\vec{v}\cdot\vec{W}/c^2}.
\end{equation}
Thus, the addition of the ternary velocity in the sense of Oziewicz can rightly be called the addition of Oziewicz--Einstein.   
\end{theorem}
\begin{proof}
In the 3D Dragan approach, the relationship between the ternary velocity (as “inverse” velocity) with the Lorentz transformation and Einstein velocities composition is not in doubt and results almost from the very definition of ternary velocity (see \cite{APP A,poster}). Oziewicz 4D approach is more general and, apart from the definition of ternary velocity, must be based on the definition of its addition in the form  (\ref{ternarne1}) or (\ref{ternarne2}). This makes the examined thesis a non-obvious theorem interpreting Einstein velocities composition in terms of ternary velocity -- since it states that 4D formulas are fully analogous to 3D formulas. The latter statement requires a formal proof by calculations.

In order to use the definition (\ref{ternarne2}) we first apply the covariant Lorentz transformation in the active version (see (\ref{La})):
\begin{equation}
  \mu^{\alpha}=La^{\alpha}_{\beta}(\sigma,\zeta)\ \nu^{\beta}=\nu^{\alpha}-\frac{\nu\circ(\zeta+\sigma)}{\sigma\circ(\zeta+\sigma)}(\zeta^{\alpha}+\sigma^{\alpha})+2\frac{\nu\circ\sigma}{c^2}\zeta^{\alpha}.
\end{equation}
The needed scalar products can be expressed as follows:
\begin{equation}
    \nu\circ\sigma=\gamma_{\nu}c^2=\gamma_{\omega}c^2, \ \ \ \ \zeta\circ\sigma=\gamma_{\zeta}c^2=\gamma_{\xi}c^2,
\end{equation}
\begin{equation}
    \nu\circ\zeta=\gamma_{\omega}\gamma_{\xi}\big(c^2-\vec{v}\cdot\vec{W}\big)=\gamma_{\omega}\gamma_{\xi}(c^2+\omega\circ\xi),
\end{equation}
\begin{equation}
    \mu\circ\sigma=2\frac{\nu\circ\sigma}{c^2}\zeta\circ\sigma-\nu\circ\zeta=\gamma_{\omega}\gamma_{\xi}(c^2-\omega\circ\xi).
\end{equation}
Now it is enough to go from Minkowski velocities to space-like (binary or ternary) velocities:  $\mu \rightarrow \mu\dsub\sigma$, $\nu=\omega\langle +|\sigma$, $\zeta=\xi\langle +|\sigma$:
\begin{equation}
    (\mu\dsub\sigma)^{\alpha}=\frac{\omega^{\alpha}+\sigma^{\alpha}}{\gamma_{\xi}\big(1-\frac{\omega\circ\xi}{c^2}\big)}-\frac{\big(\gamma_{\xi}+1+\gamma_{\xi}\frac{\omega\circ\xi}{c^2}\big)(\gamma_{\xi}\xi^{\alpha}+\gamma_{\xi}\sigma^{\alpha}+\sigma^{\alpha})}{\gamma_{\xi}\big(1-\frac{\omega\circ\xi}{c^2}\big)(\gamma_{\xi}+1)}+2\frac{\xi^{\alpha}+\sigma^{\alpha}}{1-\frac{\omega\circ\xi}{c^2}}-\sigma^{\alpha},
\end{equation}
which, after algebraic transformations, leads to the thesis of theorem. 
\qed
\end{proof}

It is worth noting that Oziewicz-Einstein addition of the ternary velocity does not depend explicitly on $\sigma^{\alpha}$, but only implicitly in the form of v. type assumptions about orthogonality: $\omega\circ\sigma=\xi\circ\sigma=(\omega\boxplus^{\sigma}_{O}\xi)\circ\sigma=0$. It should not be thought that Oziewicz's theorem only indicates the change of the language of Minkowski's time-like velocities into space-like (ternary and binary) velocities. First, the ternary velocity itself is something significantly different from Einstein's composition of velocities. This is evidenced by the fulfillment of assumption iii. in opposition to the assumption iv. -- and fulfilling ii. in opposition to the assumption i.  Moreover, ternary relative velocity is antisymmetric (assumption viii.), and Einstein composition is not antisymmetric in subtracting, nor symmetrical in adding.  Secondly, adding ternary velocity in the sense of Oziewicz--Einstein is a special case of ternary composition -- the general case will be attacked in further research.  

The concept of relative ternary velocity was published in full general by Oziewicz in 2004 \cite{Oziewicz 1} and 2007 \cite{Oziewicz 2} (see also \cite{Oziewicz 3,Oziewicz 4,Oziewicz Page}). A 3D version of this velocity for addition was published in the Ungar handbook in 2001 \cite{Ungar 2001} (see also \cite{Ungar 2005}). Definitions of this velocity were published by Ungar as early as 1992, but at that time he did not provide the result of its calculation \cite{Ungar 2}.  However, Ungar did not recognize the special importance of ternary velocity (“Einstein coaddition”) and preferred other velocities: proper velocity \cite{Ungar 2006}, M\"obius velocity addition and M\"obius coaddition \cite{Ungar 2007}. The author has already quoted in \cite{RinP} the velocity given by Ungar \cite{Ungar 2007}, but in a differentially equivalent 3D ternary velocity version, because Ungar did not provide the full formula for this velocity in this paper. 

The Lorentz transformation in terms of initial and final vectors developed along a parallel path (Wyk 1986, 1991 \cite{Wyk,Wyk 2}, also Ungar 1992 \cite{Ungar 2}). This concept is similar to the covariant Lorentz transformation, discussed later. A slightly similar approach based on the Lorentz transformation of the 4D position vector was used by Urbantke in 2002 \cite{Urbantke}. The approaches of these three researchers (Wyk, Ungar, Urbantke), however, were less general than that of Oziewicz (based on \cite{Celakoska 1,Celakoska 2}). The initial and final vectors approach (Lorentz link boost problem) was described by Celakoska in 2008 \cite{Celakoska 1} and Celakoska and Chakmakov in 2011 \cite{Celakoska 2}. 
In 2015, however, Celakoska's team published quite clear formulae for the relative 4D and 3D ternary velocities \cite{Celakoska 4}. The Lorentz link boost approach was used, but the 3D result was similar to the Urbantke formula. Unfortunately, for the subtracted 3D velocities, a different normalization was used (like the Minkowski four-velocity), which for the resulting 3D ternary velocity -- it simplified the formula, but made it difficult to read it correctly. 

On the other hand, a strictly 3D approach was used in 2012 by Dragan in resuming his lectures \cite{Dragan}. The limitation to 3D description makes it impossible to fully recognize the ternary feature of relative velocity. Nevertheless, Dragan presented a slightly more general approach than Urbantke (“magic four-rule”) in the further chapter of his lectures \cite{Dragan,Dragan WS} (see also \cite{Oziewicz eter,Oziewicz Page}).

The use of Oziewicz--Ungar--Dragan ternary velocity allowed the author to define a new concept of ternary relativistic acceleration \cite{RinP,APP A,poster,poster 2}. This acceleration solved the problem of the non-parallelism of the force to acceleration in three-dimensional space in Special Relativity. This, in turn, made it possible to define the inertial mass unifying the transverse and longitudinal mass. The use of the same method of defining acceleration for Einstein velocity composition, instead of Oziewicz--Ungar--Dragan velocity, does not solve these problems of relativistic dynamics -- see the work of Celakoska and Chakmakov \cite{Celakoska 3}. Celakoska and Chakmakov did not use the relative ternary velocity (Oziewicz-Ungar-Dragan), probably because they did not recognize the fundamental role of this velocity, which they themselves described earlier \cite{Celakoska 4}. 

The discovery of the ternarity of relative velocity in a way reflects Oziewicz's anti-Popperian philosophy of science. Contrary to Popper's postulates, Oziewicz draws attention to the need for individual and subjective scientific research conducted by specialists \cite{Oziewicz Subjective}. The dependence of the relative velocity of two bodies on the observer is a specific example of this philosophy of anti-objectivism. A subjectively preferred frame of reference \cite{Choi} may be associated with some idea of Æther  which is occasionally considered \cite{Oziewicz eter,Szostek}. It should be emphasized, however, that ternary velocity is an orthodox concept in Special Relativity.  We find some resemblance to anti-Popperian philosophy of science in the Pythagoreans, who did not apply the principles of democracy to science. However, this should be understood as freedom of research, not as monopoly on it by any group \cite{Oziewicz Reviewers}. 

\subsection{3D formulae of relative ternary velocity}

While the article is essentially written in 4D language, some things may seem (and are) simpler in terms of 3D. In the case of ternary velocity, the simplification of the 3D formula may result from the apparent loss of ternary property, which is hidden in the selection of the 3D reference frame.  In any case, it's worth knowing the explicit 3D formulas of Ungar, Dragan, Urbantke and Celakoska's team in order to be able to directly compare them with Einstein's formula. 

Let the Cartesian frame fit to the Minkowski velocity of reference frame such that  $\sigma^{\alpha}=(c,\vec{0}$). In other words, such a Cartesian frame coincides with the reference frame. Then the relative ternary velocity of the body with the velocity of Minkowski $\mu^{\alpha}=(\gamma_uc,\gamma_u\vec{u})$ in relation to the body with the velocity of Minkowski $\nu^{\alpha}=(\gamma_vc,\gamma_v\vec{v})$ is reduced to practically three-dimensional velocity  $\xi^{\alpha}=(0,\vec{W})$. The 3D ternary velocity $\vec{W}$ can be expressed in three forms differing in the way of writing the normalization coefficient.  Here is the formula given by Ungar (“Einstein coaddition”) \cite{Ungar 2001,Ungar 2005}, changed to subtraction :
\begin{equation}
\label{W Ungar}
    \vec{W}=\vec{u}\boxminus \vec{v}:=\frac{\gamma_u+\gamma_v}{\gamma_u\gamma_v(1-\vec{u}\cdot\vec{v}/c^2)+\gamma_u^2+\gamma_v^2-1}(\gamma_u\vec{u}-\gamma_v\vec{v}),
\end{equation}
which is most similar to the 4D Oziewicz formula  \cite{APP A}. Dragan's original formula is slightly different \cite{Dragan,Dragan WS,RinP}:
\begin{equation}
\label{W Dragan}
    \vec{W}=\vec{u}\ominus_{\wedge} \vec{v}:=\frac{\gamma_u^{-1}+\gamma_v^{-1}}{1-\vec{u}\cdot\vec{v}/c^2+\gamma_u\gamma_v(1-u^2v^2/c^4)}(\gamma_u\vec{u}-\gamma_v\vec{v}).
\end{equation}
However, the simplest is the formula resulting from the work Urbantke \cite{Urbantke} written for relative velocity by Celakoska's team \cite{Celakoska 4} in the form: 
\begin{equation}
\label{W Urbantke}
    \vec{W}=\vec{u}\boxminus_{\wedge}\vec{v}:=\frac{2(\gamma_u+\gamma_v)}{(\gamma_u+\gamma_v)^2+(\gamma_u\vec{u}-\gamma_v\vec{v})^2/c^2}(\gamma_u\vec{u}-\gamma_v\vec{v}).
\end{equation}
The above equation will come out if we change two four-positions into two four-velocities in the Urbantke velocity formula (see further (\ref{44}))  -- but Urbantke did not do it in \cite{Urbantke}. Oziewicz referred to the original formula of Urbantke \cite{Oziewicz eter,Oziewicz Page}. Ungar also defined a formula of the Urbantke type, but did not write it out fully explicitly  \cite{Ungar 2001}. As already mentioned, a slightly more general formula than Urbantke was proposed by Dragan, which he called “magic four-rule”  \cite{Dragan,Dragan WS}. However, strictly speaking, this is not a strictly 4D rule, nor strictly 3D, but it describes the correspondence between 4D and 3D (and \textit{vice versa}). 

\begin{theorem}[About 3D ternary velocity formulae]
\label{Theorem W}
The formulae of Ungar (\ref{W Ungar}), Dragan (\ref{W Dragan}) and Urbantke-Celakoska (\ref{W Urbantke}) of the ternary velocity $\vec{W}$ are identically equal ($\boxminus\equiv\ominus_{\wedge}\equiv\boxminus_{\wedge}$). These formulae for ternary relative velocity $\vec{W}$ express the velocity of the active Loretnz boost “accelerating” the velocity $\vec{v}$ to velocity $\vec{u}$: $La_{\vec{W}}(\vec{v})=\vec{u}$, or express the velocity of the passive Lorentz boost “delaying” the velocity of $\vec{u}$ to the velocity of $\vec{v}$: $Lp_{\vec{W}}(\vec{u})=\vec{v}$. In other words, ternary velocity is a parameter of the ordinary passive Lorentz boost $Lp_{\vec{W}}[\vec{r}]_t=\vec{r}'$.
\end{theorem}

\begin{proof}
We'll start by transforming the denominator of formula (\ref{W Urbantke}):
\begin{equation}
\begin{split}
    (\gamma_u+\gamma_v)^2+\frac{(\gamma_u\vec{u}-\gamma_v\vec{v})^2}{c^2}
    =\gamma_u^2+\gamma_v^2+2\gamma_u\gamma_v+\gamma_u^2\frac{u^2}{c^2}+
    \gamma_v^2\frac{v^2}{c^2}-2\gamma_u\gamma_v\frac{\vec{u}\cdot\vec{v}}{c^2}\\
    =\gamma_u^2+\gamma_v^2+2\gamma_u\gamma_v+\gamma_u^2(1-\gamma_u^{-2})+
    \gamma_v^2(1-\gamma_v^{-2})-2\gamma_u\gamma_v\frac{\vec{u}\cdot\vec{v}}{c^2}\\
    =2\gamma_u^2+2\gamma_v^2-2 +2\gamma_u\gamma_v\Big(1-\frac{\vec{u}\cdot\vec{v}}{c^2}\Big),
\end{split}
\end{equation}
which is equal to double the denominator of the formula (\ref{W Ungar}). Since the numerators also differ twice, the formulas (\ref{W Ungar}) and (\ref{W Urbantke}) are indentity equal.   

Now let's transform the last term of the formula denominator of (\ref{W Dragan}):
\begin{equation}
    \gamma_u\gamma_v(1-u^2v^2/c^4)
    =\gamma_u\gamma_v-\gamma_u\gamma_v(1-\gamma_u^{-2})(1-\gamma_v^{-2})
    =\gamma_u\gamma_v^{-1}+\gamma_u^{-1}\gamma_v-\gamma_u^{-1}\gamma_v^{-1}.
\end{equation}
If we take this into account and extend the numerator and denominator of the formula (\ref{W Dragan}) with the factor $\gamma_u\gamma_v$, we get the formula (\ref{W Ungar}).  So the formulae (\ref{W Ungar}),(\ref{W Dragan}),(\ref{W Urbantke}) are indentity equal.

Now consider the passive Lorenzt transformation, written for the direction parallel and perpendicular to the direction of velocity $\vec{W}$:
\begin{equation}
    \vec{r}'_{\parallel}=\gamma_W\big(\vec{r}_{\parallel}-\vec{W}t\big) \:\:\:,\:\:\:
    \vec{r}'_{\perp}=\vec{r}_{\perp}
    \:\:\:,\:\:\:
    t'=\gamma_W\big(t-\vec{W}\cdot\vec{r}/c^2\big).
\end{equation}
The spatial part of the transformation can be expressed by one equation: 
\begin{equation}
    \vec{r}'=\vec{r}-\gamma_W\vec{W}t+\frac{\gamma_W^2}{\gamma_W^2+1}\frac{\vec{r}\cdot\vec{W}}{c^2}\vec{W}.
\end{equation}
If the dot product in this equation is expressed in terms of times, then from the equation: 
\begin{equation}
    \vec{r}'=\vec{r}-\gamma_W\vec{W}t+\frac{\gamma_W^2}{\gamma_W^2+1}\Big(t-\frac{t'}{\gamma_W}\Big)\vec{W},
\end{equation}
the velocity can be easily calculated: 
\begin{equation}
    \vec{W}=(1+\gamma_W^{-1})\frac{\vec{r}-\vec{r}'}{t+t'}.
\end{equation}
We still have to compute $\gamma_W$. For this purpose, we use the relationship  $W^2/c^2=1-\gamma_W^{-2}$:
\begin{equation}
   (1+\gamma_W^{-1})^2\Big(\frac{\vec{r}-\vec{r}'}{t+t'}\Big)^2/c^2=1-\gamma_W^{-2}.
\end{equation}
The above equation after algebraic simplification becomes a linear equation for $\gamma_W$, the solution of which is: 
\begin{equation}
\gamma_W=\frac{1+\big(\frac{\vec{r}-\vec{r}'}{t+t'}\big)^2/c^2}{1-\big(\frac{\vec{r}-\vec{r}'}{t+t'}\big)^2/c^2}
\:\:\:\longleftrightarrow\:\:\:
1+\gamma_W^{-1}=\frac{2}{1+\big(\frac{\vec{r}-\vec{r}'}{t+t'}\big)^2/c^2}.
\end{equation}
Using this result, we obtain the Urbantke formula  \cite{Urbantke}:
\begin{equation}
\label{44}
    \vec{W}=\frac{2}{1+\big(\frac{\vec{r}-\vec{r}'}{t+t'}\big)^2/c^2}\: \frac{\vec{r}-\vec{r}'}{t+t'}.
\end{equation}
This formula expresses the Lorentz transformation velocity -- the most fundamental relative velocity we can imagine -- Einstein's transformed relative velocity is a velocity secondary to that (physically and mathematically, but not historically). 

In the equation (\ref{44}) the components of four-position vectors can be replaced by components of any four-vector.  So if we change the four-position $(ct,\vec{r})$ to the four-velocity $(\gamma_uc,\gamma_u\vec{u})$, and the four-position $(ct',\vec{r}')$ to the four-velocity $(\gamma_vc,\gamma_v\vec{v})$ we get the formula for the ternary velocity (\ref{W Urbantke}). Thus, in fact, this formula describes a passive Lorenzt transformation  $Lp_{\vec{W}}(\vec{u})=\vec{v}$. And the active transformation obviously works in the opposite direction $La_{\vec{W}}(\vec{v})=\vec{u}$ which was the last thing to prove. \qed
\end{proof}

\section{Explicitly covariant rewriting of Einstein composition of 4D Minkowski velocities in the ternary sense}

The most typical Lorentz transformation is the so-called Lorentz boost. This boost in the 3D formalism is parameterized with three-dimensional velocity. In addition, the 3D velocity itself is subject to other transformation laws resulting from the Lorentz transformation (Einstein velocities composition).  The Lorentz boost formula for coordinates and the 3D velocities composition formula are, however, significantly different despite their formal closeness. The consequence of this difference is the rotation of Tomas--Wigner accompanying the composition of non-parallel 3D velocities, as well as the lack of associativity for the composition velocities operation -- although Lorentz boosts composition is associative by definition of the Lorentz group 
(however, it no longer has to be considered a pure boost).

If the 3D velocities composition formula is written in 4D form, the above duality (Lorentz boost of coordinates vs velocities composition) disappears in a way. This is how the covariant Lorentz transformation (for the Lorentz boost case) will be derived in this work. Of course, the covariant Lorentz transformation can be introduced differently, but the author succeeded precisely on the basis of the 4D generalization of Einstein velocities composition. Only then was it possible to find an analogous formula for the covariant Lorentz transformation in literature. Unfortunately, it is not widely used in the literature and, apart from the source \cite{Tsamparlis} and pioneering works from about half a century ago \cite{Bazanski,Krause 1,Krause 2,Fahnline}, it is difficult to indicate many works containing the covariant Lorentz transformation -- let alone its usage.   A distractor that hinders the recognition of the covariant Lorentz transformation is its ternary character (dependence on the Minkowski velocity of the distinguished frame of reference $\sigma^{\alpha}$). It was this distractor that disrupted the author in finding the covariant Lorentz transformation while writing his master thesis \cite{mgr}. However, the thesis \cite{mgr} contained an almost covariant form of the Lorentz transformation, and only the introduction of the $\sigma^{\alpha}$ vector was missing  (e.g. $\sigma\circ U/c$ in place of $U_0$). 

We'll start by writing the full 3D Einstein  velocities composition formula. Both in the context of relative velocity and in the context of the passive Lorentz transformation, such a law is naturally a subtraction rather than an addition of velocity (contrary to Einstein's original \cite{Einstein}):
\begin{equation}
\label{Einstein minus}
    \vec{u}\ominus_{0}\vec{v}:=\vec{u}'=\frac{\vec{u}-\vec{v}+\frac{\gamma_{v}}{\gamma_{v}+1}(\vec{u}\times\vec{v})\times\vec{v}/c^2}{1-\vec{u}\cdot\vec{v}/c^2}=:\vec{\mathcal{E}}(\vec{v},\vec{u}),
\end{equation}
where $\vec{\mathcal{E}}$ is the relative velocity of a body moving at velocity $\vec{u}$ relative to a body moving at velocity $\vec{v}$.
 Now we have to rewrite the above law using space-like 4D (binary) relative velocities.  Let $\vec{v}\rightarrow \nu^{\mu} \rightarrow \omega^{\alpha}_{01}=(\nu\dsub\sigma)^{\alpha}$ and $\vec{u}\rightarrow\mu^{\alpha}\rightarrow\omega^{\alpha}_{02}=(\mu\dsub\sigma)^{\alpha}$. Based on the identity $(\vec{u}\times\vec{v})\times\vec{v}=(\vec{u}\cdot\vec{v})\vec{v}-v^2\vec{u}$ the formula (\ref{Einstein minus}) can be written as follows:
\begin{equation}
\label{Einstein 4D}
    (\omega_{02}\boxminus_{0}^{\sigma}\omega_{01})^{\alpha}:=\frac{\omega^{\alpha}_{02}-\omega^{\alpha}_{01}+\frac{\gamma_{01}}{\gamma_{01}+1}\Big(\frac{\omega_{01}\circ\omega_{01}}{c^2}\omega^{\alpha}_{02}-\frac{\omega_{02}\circ\omega_{01}}{c^2}\omega^{\alpha}_{01}\Big)}{1+\omega_{02}\circ\omega_{01}/c^2}=:\varepsilon^{\alpha}(\omega_{01},\omega_{02}),
\end{equation}
where $\boxminus^{\sigma}_0$ is a 4D generalization from the definition of a 3D operation $\ominus_0$. Due to the signature selection, the signs before the 4D scalar products have been changed compared to the 3D scalar products ($\vec{u}\cdot\vec{v}\rightarrow-\omega_{02}\circ\omega_{01}$).Thus, Einstein velocities composition is treated here as some composition of binary velocities -- but other than transitive Oziewicz composition of binary velocities.  We already know from the previous section that the Lorentz boost velocity $\omega_{01}^{\alpha}$ can be interpreted up to the sign as ternary velocity. It can be concluded that the inverse implication to Oziewicz's theorem is derived now. Thus, the form of $\boxminus_0^{\sigma}$ (and consequently $\boxplus_0^{\sigma}:=-\boxminus_0^{\sigma}$) is a definition, and Oziewicz's theorem was the thesis that $\boxplus_O^{\sigma}=\boxplus_0^{\sigma}$. In other words, the Oziewicz operation  $\boxplus_{O}^ {\sigma}$ (Oziewicz--Eisntein operation) was defined differently than the usual 4D generalization   $\boxplus_0^{\sigma}$ of 3D Einstein operation $\oplus_0:=-\ominus_0$. In any case, using the action $\boxminus_o^{\sigma}$ we obtained the space-like velocity $\varepsilon^{\alpha}$, which will be called the relative velocity of Einstein--Oziewicz.  The Einstein--Oziewicz velocity $\varepsilon^{\alpha}$ depends explicitly only on two space-like composition (subtracted) velocities, which are orthogonal to the Minkowski velocity $\sigma^{\alpha}$ of reference frame. 

Now we will express the  Einstein--Oziewicz relative velocity $\varepsilon^{\alpha}$ using the Minkowski velocities $\nu^{\alpha}, \mu^{\alpha}$, as well as $\sigma^{\alpha} $, which will directly reflect the ternary character of the Einstein--Oziewicz relative velocity (see Fig. \ref{ternaryfig}). 

\begin{theorem}
The space-like Einstein-Oziewicz relative velocity  $\varepsilon^{\alpha}$ of the material point $\mu^{\alpha}$ in relation to the material point $\nu^{\alpha}$ in the reference frame of velocity $\sigma^{\alpha}$ is expressed as follows in terms of the three above-mentioned Minkowski velocities: 
\begin{equation}
\label{teza 3}
    \varepsilon^{\alpha}(\sigma,\nu,\mu)=\frac{c^2}{\mu\circ\nu}\Big[\mu^{\alpha}-\frac{\mu\circ(\nu+\sigma)}{\sigma\circ(\nu+\sigma)}(\nu^{\alpha}+\sigma^{\alpha})\Big]+\sigma^{\alpha}=:(\mu [-]_0^{\sigma} \nu)^{\alpha},
\end{equation}
which also specifies a new type (ternary or pseudo-binary) of Minkowski velocities subtraction $[-]_0^{\sigma}$ with a space-like value.
\end{theorem}

\begin{proof}
Binary velocities appearing in the formula (\ref{Einstein 4D}) have to be written into the Minkowski velocities  $\omega^{\alpha}_{01}=(\nu\dsub\sigma)^{\alpha}$ and $\omega^{\alpha}_{02}=(\mu\dsub\sigma)^{\alpha}$:
\begin{equation}
    \varepsilon^{\alpha}=\frac{\frac{c^2}{\mu\circ\sigma}\mu^{\alpha}-\frac{c^2}{\nu\circ\sigma}\nu^{\alpha}+\frac{\nu\circ\sigma}{\nu\circ\sigma+c^2}\Big(\frac{\omega_{01}\circ\omega_{01}}{c^2}\big(\frac{c^2}{\mu\circ\sigma}\mu^{\alpha}-\sigma^{\alpha}\big)-\frac{\omega_{02}\circ\omega_{01}}{c^2}\big(\frac{c^2}{\nu\circ\sigma}\nu^{\alpha}-\sigma^{\alpha}\big)\Big)}{1+\omega_{02}\circ\omega_{01}/c^2}.
\end{equation}
And the scalar products are equal to:
\begin{equation}
    \omega_{01}\circ\omega_{01}=-v^2=\frac{c^6}{(\nu\circ\sigma)^2}-c^2, 
\end{equation}
\begin{equation}
    \omega_{02}\circ\omega_{01}=\frac{(\mu\circ\nu)c^4}{(\mu\circ\sigma)(\nu\circ\sigma)}-c^2. 
\end{equation}
We can simplify the computational complexity by introducing auxiliary symbols:  $x:=\nu\circ\sigma/c^2,y:=\mu\circ\sigma/c^2, z:=\mu\circ\nu/c^2$. Then:
\begin{equation}
    \varepsilon^{\alpha}=\frac{\frac{1}{y}\mu^{\alpha}-\frac{1}{x}\nu^{\alpha}+\frac{x}{x+1}\Big(\frac{1-x^2}{x^2}\big(\frac{1}{y}\mu^{\alpha}-\sigma^{\alpha}\big)-\frac{z-xy}{xy}\big(\frac{1}{x}\nu^{\alpha}-\sigma^{\alpha}\big)\Big)}{1+\frac{z-xy}{xy}},
\end{equation}
which after transformation gives the expression:
\begin{equation}
    \varepsilon^{\alpha}=\frac{\mu^{\alpha}-\frac{y+z}{x+1}\nu^{\alpha}+\frac{xz-y}{x+1}\sigma^{\alpha}}{z},
\end{equation}
equivalent to thesis of theorem. \qed
\end{proof}

The space-like 4D relative velocity of Einstein--Oziewicz can now be assigned to the 4D time-like relative velocity (Minkowski type): 
\begin{equation}
\label{def beta}
    \beta^{\alpha}(\sigma,\nu,\mu):=\big(\varepsilon(\sigma,\nu,\mu) \ \langle + | \ \sigma\big)^{\mu}=\frac{\varepsilon^{\mu}+\sigma^{\mu}}{\sqrt{1+\varepsilon\circ\varepsilon/c^2}}=:(\mu\boxminus^{\sigma}_M\nu)^{\alpha}.
\end{equation}
The velocity of $\beta^{\alpha}$ will be called the Einstein-Minkowski relative velocity. Note that the operation $\boxminus^{\sigma}_M$ which we will call (ternary or pseudo-binary) Minkowski velocities subtraction is an internal operation. Explicitly calculating the Einstein--Minkowski velocity is not so difficult now, but it is related to the simplest form of the covariant Lorentz transformation -- therefore it will be given in the form of the following theorem:

\begin{theorem}[Covariant passive Lorentz transformation of velocity]
The time-like Einstein--Minkowski relative velocity $\beta^{\alpha}$ is a passive Lorentz transformation (boost) of the Minkowski velocity $\mu^{\alpha}$ to a system with the Minkowski velocity $\nu^{\alpha}$ from the reference system with the Minkowski velocity of $\sigma^{\alpha}$, and has the following form: 
\begin{equation}
\label{beta}
    \beta^{\alpha}(\sigma,\nu,\mu)=(\mu \boxminus_M^{\sigma} \nu)^{\alpha}=Lp^{\alpha}_{\beta}(\sigma,\nu)\mu^{\beta}=\mu^{\alpha}-\frac{\mu\circ(\nu+\sigma)}{\sigma\circ(\nu+\sigma)}(\nu^{\alpha}+\sigma^{\alpha})+2\frac{\mu\circ\nu}{c^2}\sigma^{\alpha}.
\end{equation}
\end{theorem}

\begin{proof}
In order to apply the definition (\ref{def beta}) we first calculate the square of the Einstein--Oziewicz velocity norm:
\begin{equation}
    \varepsilon\circ\varepsilon=\frac{c^6}{(\mu\circ\nu)^2}-c^2,
\end{equation}
which is the same as for the  Oziewicz--Świerk--Bol\'os binary velocity and with the accuracy to the sign as for the Einstein velocities composition (see Statement 2 in \cite{APP A}). Since $\mu\circ\nu>0$ then applying (\ref{def beta}) to (\ref{teza 3}) we directly obtain the last formula (\ref{beta}) of the thesis according to  \cite{Krause 2}.

It remains to be explained that the Einstein--Minkowski velocity is a Lorentz boost of velocity $\mu^{\alpha}$ into the system with the velocity $\nu^ {\alpha}$ (passive transformation $Lp$).  Indeed it is so, because Einstein--Minkowski velocity is a time-like 4D notation analogous to velocity subtraction (\ref{Einstein minus}) or (\ref{Einstein 4D}). Nevertheless, it is still worth making sure that we are dealing with a passive (and not active) transformation, for which $\mu^{\alpha}=\nu^{\alpha}$ does not change:
\begin{equation}
    \beta^{\alpha}(\sigma,\nu,\nu)=Lp^{\alpha}_{\beta}(\sigma,\nu)\ \nu^{\beta}=\nu^{\alpha}
\end{equation}
The condition of identity transformation also occurs when the passive Lorentz boost does not change the main (selected) frame of reference  $\nu^{\alpha}=\sigma^{\alpha}$:
\begin{equation}
    \beta^{\alpha}(\sigma,\sigma,\mu)=Lp^{\alpha}_{\beta}(\sigma,\sigma)\ \mu^{\beta}=\mu^{\alpha}.
\end{equation}
On the other hand, a passive boost of nominally zero velocity $\sigma^{\alpha}$ with velocity $\nu^{\alpha}$ gives the opposite velocity in a three-dimensional sense  $\tilde{\nu}^{\alpha}=(\gamma_{v}c,-\gamma_v\vec{v})$: 
\begin{equation}
\label{tilda}
    \beta^{\alpha}(\sigma,\nu,\sigma)=Lp^{\alpha}_{\beta}(\sigma,\nu)\ \sigma^{\beta}=2\frac{\sigma\circ\nu}{c^2}\sigma^{\alpha}-\nu^{\alpha}=\tilde{\nu}^{\alpha},
\end{equation}
when active transformation would give ordinary velocity $\nu^{\alpha}$. This fact concludes the proof of the theorem.
\qed
\end{proof}

On the basis of the covariant passive Lorentz transformation, one can easily structure the inverse transformation and thus the active transformation. Despite the simplicity of calculations, it is worth presenting these properties in the form of a theorem. 

\begin{theorem}[Covariant active and inverse Lorentz transformations]
If the Lorentz boost is parameterized with two Minkowski velocities $\sigma^{\alpha}$ and $\nu^{\alpha}$, then swapping these velocities leads to an inverse transformation or equivalently converts the passive transformation $Lp$ into the active transformation $La$ (and vice versa):  
\begin{equation}
\label{T5i}
    Lp^{-1}(\sigma,\nu)=Lp(\nu,\sigma)=La(\sigma,\nu),
\end{equation}
\begin{equation}
\label{T5ii}
    La^{-1}(\sigma,\nu)=La(\nu,\sigma)=Lp(\sigma,\nu).
\end{equation}
 The active covariant Lorentz transformation of Minkowski velocity explicitly takes the following form: 
\begin{equation}
\label{La}
    La^{\alpha}_{\beta}(\sigma,\nu)\ \mu^{\beta}=\mu^{\alpha}-\frac{\mu\circ(\nu+\sigma)}{\sigma\circ(\nu+\sigma)}(\nu^{\alpha}+\sigma^{\alpha})+2\frac{\mu\circ\sigma}{c^2}\nu^{\alpha},
\end{equation}
which differs from the passive transformation only in the last expression.
\end{theorem}

\begin{proof}
The inverse Lorentz transformation is obtained by using a boost with the velocity $\tilde{\nu}^{\alpha}$ opposite in the three-dimensional sense to  $\nu^{\alpha}$:
\begin{equation}
    (Lp^{-1})^{\alpha}_{\beta}(\sigma,\nu) \ \mu^{\beta}=Lp^{\alpha}_{\beta}(\sigma,\tilde{\nu})\ \mu^{\beta}=\mu^{\alpha}-\frac{\mu\circ(\tilde{\nu}+\sigma)}{\sigma\circ(\tilde{\nu}+\sigma)}(\tilde{\nu}^{\alpha}+\sigma^{\alpha})+2\frac{\mu\circ\tilde{\nu}}{c^2}\sigma^{\alpha}.
\end{equation}
Using (\ref{tilda}) and symbols from the proof of theorem 3, we write further the above expression: 
\begin{equation}
    (Lp^{-1})^{\alpha}_{\beta} \mu^{\beta}=\mu^{\alpha}-\frac{2xy-z+y}{2x-x+1}(2x\sigma^{\alpha}-\nu^{\alpha}+\sigma^{\alpha})+2(2xy-z)\sigma^{\alpha},
\end{equation}
which after transformations gives the formula: 
\begin{equation}
    (Lp^{-1})^{\alpha}_{\beta} \mu^{\beta}=\mu^{\alpha}-\frac{z+y-2(x+1)y}{x+1}\nu^{\alpha}-\frac{z+y}{x+1}\sigma^{\alpha},
\end{equation}
which is equal to the right side of the thesis (\ref{La}) of the theorem.  The inverse of the Lorentz transformation automatically is the conversion of the passive transformation to the active transformation (and \textit{vice versa}). Therefore, the third thesis (\ref{La}) is proved. 

Now it is enough to note that the proven expression (\ref{La}) differs from (\ref{beta}) only by swapping  $\sigma^{\alpha}$ and $\nu^{\alpha}$. This proves (\ref{T5i}), (\ref{T5ii}) and ends the proof of the theorem.   
\qed
\end{proof}

The first conclusion from the above theorem is the possibility of determining the internal addition in the Minkowski velocities set: 
\begin{equation}
    (\mu \boxplus_M^{\sigma} \nu)^{\alpha}=La^{\alpha}_{\beta}(\sigma,\nu)\ \mu^{\beta}=\mu^{\alpha}-\frac{\mu\circ(\nu+\sigma)}{\sigma\circ(\nu+\sigma)}(\nu^{\alpha}+\sigma^{\alpha})+2\frac{\mu\circ\sigma}{c^2}\nu^{\alpha}.
\end{equation}
The Minkowski velocities subtraction  described earlier results the  Einstein--Minkowski relative velocity $\beta^{\alpha}$, but result of the addition  appears to be something different velocity. Nevertheless, this result of addition will not be separately named here and denoted by a new symbol. Surprisingly, however, it can be interpreted as a relative velocity, but relative to another body in a different frame of reference: 
\begin{equation}
    (\mu \boxplus_M^{\sigma} \nu)^{\alpha}=(\mu \boxminus_M^{\nu} \sigma)^{\alpha}=\beta^{\alpha}(\nu,\sigma,\mu).
\end{equation}
The above relation shows the flexibility and importance of the concept of Einstein--Minkowski relative velocity in the context of Ternary Special Relativity. 

An analogous second conclusion of the theorem will be the definition of addition as opposed to subtraction that creates a space-like Einstein--Oziewicz velocity (expressed in terms of Minkowski velocities):
\begin{equation}
    (\mu [-]_0^{\sigma} \tilde{\nu})^{\alpha}=:(\mu [+]_0^{\sigma} \nu)^{\alpha}:=\big(\beta(\nu,\sigma,\mu)\dsub\sigma\big)^{\alpha}.
\end{equation}
Both ways of defining (on the left and on the right) are equivalent and lead to the formula: 
\begin{equation}
    (\mu [+]_0^{\sigma} \nu)^{\alpha}=\frac{c^4}{2(\nu\circ\sigma)(\mu\circ\sigma)-(\mu\circ\nu) c^2}\Big[\mu^{\alpha}-\frac{\mu\circ(\nu+\sigma)}{\sigma\circ(\nu+\sigma)}(\nu^{\alpha}+\sigma^{\alpha})+2\frac{\mu\circ\sigma}{c^2}\nu^{\alpha}\Big]-\sigma^{\alpha}.
\end{equation}
The above formula, given without proof, no longer shows the simplicity and symmetry of the previous formulas.  However, the Einstein--Oziewicz velocity expressed by the operation of Oziewicz--Einstein on binary velocities has a trivial difference between addition and subtraction:
\begin{equation}
    (\omega_{01} \boxplus_0^{\sigma} \omega_{02})^{\alpha}\equiv(\omega_{01} \boxplus_O^{\sigma} \omega_{02})^{\alpha}:=(\omega_{01} \boxminus_0^{\sigma} -\omega_{02})^{\alpha}\equiv(\omega_{01} \boxminus_O^{\sigma} -\omega_{02})^{\alpha}.
\end{equation}

At the end of the work, the inverse velocity theorem for the Lorentz boost, which is the key Ternary Special Relativity theorem, will be formulated and proven. It is not about the opposite Minkowski velocity $\tilde{\nu}^{\alpha}$ generating the inverse Lorentz transformation, but the inverse velocity in terms of a parameter for a given Lorentz boost. The term inverse velocity is borrowed from Oziewicz, who used it both in the context of binary relative velocity and in the context of the Lorentz transformation leading to ternary velocity.

\begin{theorem}[Inverse velocity of the covariant Lorentz transformation]
\label{Theorem inverse velocity}
If a passive Lorentz boost $Lp(\sigma,\nu)$ converts Minkowski velocity $\mu^{\alpha}$ to velocity $\beta^{\alpha}$, then the parameter of this boost, being Minkowski velocity, is expressed as follows :
\begin{equation}
\label{inverse}
    Lp^{\alpha}_{\beta}(\sigma,\nu) \mu^{\beta}=\beta^{\alpha} \ \rightarrow \ \nu^{\alpha}(\sigma,\beta,\mu)=\frac{(\mu\circ\sigma+\beta\circ\sigma)(\mu^{\alpha}-\beta^{\alpha})+\big(2\frac{(\beta\circ\sigma)^2}{c^2}+\mu\circ\beta-c^2\big)\sigma^{\alpha}}{2(\mu\circ\sigma)(\beta\circ\sigma)/c^2-\mu\circ\beta+c^2},
\end{equation}
and can be treated as some kind of subtraction of $\mu^{\alpha}$ and $\beta^{\alpha}$, just as $\beta^{\alpha}$ is a specific subtraction of $\mu^{\alpha}$ and $\nu^{\alpha}$. 
\end{theorem}

\begin{proof}
We need to solve the equation (\ref{beta}) for $\beta^{\alpha}$ treating $\nu^{\alpha}$ as the unknown. Starting from the general ternary formula  (\ref{general ternary}):
\begin{equation}
\label{anzac}
    \nu^{\alpha}(\sigma, \beta, \alpha)=f_2 \mu^{\alpha}-f_1 \beta^{\alpha}-f_0\sigma^{\alpha},
\end{equation}
it is worth noting that $f_2=f_1$ (condition viii. holds). The justification for this fact is the simplest in 3D.  It is known that if $Lp_{\vec{v}}(\vec{u})=\vec{b}$ then $Lp_{-\vec{v}}(\vec{b})=\vec{u}$. This means the antisymmetry of the three-dimensional inverse velocity $\vec{v}(\vec{b},\vec{u})=-\vec{v}(\vec{u},\vec{b})$, which is equivalent to $f_2=f_1$.

So it is enough to write the equations for the two unknown coefficients $f_0$ and $f_1$. For this, we calculate the two scalar products resulting from the formula $\beta^{\alpha}$ (\ref{beta}):
\begin{equation}
    \beta\circ\sigma=\mu\circ\nu,
\end{equation}
\begin{equation}
    \beta\circ\nu=2(\mu\circ\nu)(\nu\circ\sigma)/c^2-\mu\circ\sigma.
\end{equation}
By inserting the formula (\ref{anzac}) into these two equations and using simplifying symbols  $b:=\beta\circ\sigma/c^2, m:=\mu\circ\sigma/c^2, n:=\mu\circ\beta/c^2$ we get:
\begin{equation}
    b=(1-n)f_1-m \ f_0,
\end{equation}
\begin{equation}
    (n-1)f_1-b \ f_0=2b\big((m-b)f_1-f_0\big)-m.
\end{equation}
The solution of this linear system of equations (e.g. by the determinant method) leads, after simplification, to the result: 
\begin{equation}
    f_1=\frac{b+m}{1-n+2bm},
\end{equation}
\begin{equation}
    f_0=\frac{1-n-2b^2}{1-n+2bm},
\end{equation}
which is equivalent to the thesis of the theorem. 
\qed
\end{proof}

In the case of the inverse velocity for the active transformation $La(\sigma,\nu)$ it is enough to replace the arguments in the resulting formula  $\nu^{\alpha}(\sigma,\beta,\mu)\longrightarrow \nu^{\alpha}(\sigma,\mu,\beta)$.

The Oziewicz--Ungar--Dragan ternary relative velocity from the previous section is based on the above theorem. The application of the covariant Lorentz transformation allowed for a significant simplification of the proof of the theorem. For example, in Dragan \cite{Dragan}, the calculations extend through three subsections (it does not apply to the “magic four-rule”), and the previous relatively simple heuristic derivation of the author \cite{APP A} was not fully general and had to be based on a computationally complex  proof of the relevant lemma. However, the 3D derivation presented in the proof of Theorem \ref{Theorem W} turned out to be the simplest.



\subsubsection{Conclusions of this work.}

The only covariant relative velocity that does not depend on the selected frame of reference  is Oziewicz--Świerk--Bolós (and Matolcsi or Bini--Carini--Jantzen) velocity (see summary Tab. \ref{velocities}). In the sense of space-like velocity, it is a standard velocity, because it is normalized to the real speed of relative motion (Minkowski time-like 4D velocity is normalized to the speed of light). Nevertheless, the composition of binary velocities has a ternary character, i.e. it additionally depends on the selected reference system.  It turns out that virtually every comprehensive relation of motion is ternary, starting with the Lorentz transformation of velocity. The truly covariant Lorentz transformation leads to a space-like relative velocity of Einstein--Oziewicz, which is ternary in nature (or at least pseudo-binary). Also the time-like equivalent of this velocity (called the Einstein--Minkowski velocity) is characterized by a certain ternarity. 

However, the role of the canonical ternary velocity is played by the relative velocity of Oziewicz--Ungar--Dragan (also Celakoska--Chakmakov--Petrushevski and Urbantke), which is based on the so-called inverse velocity in terms of the covariant Lorentz transformation parameter. The relative ternary velocity of Oziewicz--Ungar--Dragan is the only one of the considered ones to be characterized by antisymmetry, which seems obligatory for relative velocities. Even Oziewicz-Minkowski velocity, which is a time-like equivalent of canonical ternary velocity, does not have this property. Composition of ternary velocities in the Oziewicz sense takes the form of Einstein velocity composition generalized to 4D. Surprisingly, such folding is not ternary -- it does not depend on the third parameter, other than on the two folding velocities. However, this Oziewicz--Einstein ternary composition is not a fully general case of composition of relative ternary velocities based on transitivity. This fact means that Oziewicz's research program, called here Ternary Special Relativity (TSR), is rich enough and still open.

Another open issue for TSR remains the form of associativity composition of velocities. As an example of such composition, the clarified  composition of binary velocities in the Oziewicz sense, based on transitivity, was given.

\begin{table}[h!]
\centering
\caption{List of various 4D relative velocities of a body with Minkowski velocity $\mu^{\alpha}$ relative to a body with Minkowski velocity $\nu^{\alpha}$, in a frame of reference with Minkowski velocity $\sigma^{\alpha}$. }
\begin{tabular}{l c c c c}
\hline
\hline
Velocity name & Designation & Direction  & Type & Boost\\
\hline	
Oziewicz--Świerk--Bolós & $\omega^{\alpha}(\nu,\mu)$ &  space-like & binary & special parameter\\
Oziewicz--Ungar--Dragan & $\xi^{\alpha}(\sigma,\nu,\mu)$ &  \ \ space-like & \ \ canonical ternary & \ \ general parameter\\
Oziewicz--Minkowski$^*$ & $\zeta^{\alpha}(\sigma,\nu,\mu)$ &  \ \ time-like & \ \ ternary & \ \ general parameter\\
Einstein--Oziewicz$^*$ & $\varepsilon^{\alpha}(\sigma,\nu,\mu)$ &  \ \ space-like & \ \ pseudo binary & \ \ not established\\
Einstein--Minkowski & $\beta^{\alpha}(\sigma,\nu,\mu)$ &  \ \ time-like & \ \ pseudo binary & \ \ value\\
\hline
\hline
$^*$ Original authorial velocities
\end{tabular}
\label{velocities}
\end{table}

\subsubsection{Further research.}

The development of the foundations of Oziewicz's research program (Ternary Special Relativity -- TRS) turned out to be so rich that the results of a slightly different nature had to be moved to next publication. Nevertheless, the continuation will refer to the formalism defined in this work.

In continuation of this work will introduce two 4D ternary generalizations of relative binary velocity called pseudo-binary relative velocities: cross and axial. The first pseudo-binary cross velocity is a proper 4D generalization of the author's 3D jet velocity (3D binary velocity). The second is a 4D generalization of the 3D axial velocity of Fernándeuz-Guasti and the author. The cross velocity is a relatively simple modification of the binary velocity, while the second pseudo-binary axial velocity is a bit more complicated ‒ almost like typical ternary velocity. Despite this complexity, axial velocity is a straightforward direct generalization of Einstein composition of velocities in one spatial dimension into the general case of 4D space-time. Additionally, the problem of general composition of ternary velocities will be attacked -- more general than composition in the sense of Oziewicz--Einstein. 

The announced results have been written and calculated in handwritten notes and probably will be presented at the next GOL conference in memory of Zbigniew Oziewicz. If the conference is adjourned, these results will be published differently.

\subsubsection{Acknowledgments.}
Many thanks to Ryszard Kostecki for first showing me Oziewicz's works, thanks to which I was able to establish correspondence with Oziewicz. I would like to thank Dariusz Świerk for finding his master thesis after 33 years, written under the supervision of Oziewicz and containing binary velocity. In addition, I am grateful to Bill Page for the polemical discussion of the novelty of ternary velocity. Finally, thanks to Larissa Sbitneva for discussing the authorship of the first discovery of the non-associativity of the velocity composition.

%
%

\end{document}